 \documentclass[letterpaper, 10 pt, conference]{ieeeconf}  

\IEEEoverridecommandlockouts                              

\usepackage{graphics}
\usepackage{floatrow}
\usepackage{epsfig} 
\usepackage{times} 
\usepackage{amsmath} 
\usepackage{amssymb}  
\usepackage{mathrsfs}
\usepackage{multicol}
\usepackage{subcaption}
\usepackage{graphicx}
\usepackage{cite}

\newtheorem{theorem}{Theorem}

\newtheorem{lemma}{Lemma}
\newtheorem{remark}{Remark}
\newtheorem{proposition}{Proposition}
\newtheorem{definition}{Definition}
\newtheorem{assumption}{Assumption}
\usepackage{color,array}
\usepackage{algorithm,algpseudocode}
\usepackage{soul}
\usepackage{todonotes}
\title{\LARGE \bf {Optimal Hiding with Partial Information of the Seeker's Route}}

\author{Prajakta Surve$^{a,1}$
\and Shaunak D.~Bopardikar$^{b,1}$ \and 
Daigo Shishika$^{c,2}$ \and Dipankar Maity$^{d,3}$  \and  Michael Dorothy$^{e,4}$
 \thanks{*This research was supported under the grant number DCIST CRA W911NF-17-2-0181.
 DISTRIBUTION STATEMENT  A. Distribution is unlimited. 
 The views expressed in this paper are those of the authors and do not reflect the official policy or position of the United States Government, Department of Defense, or its components.}
 \thanks{$^{a}${\tt\small survepra@msu.edu}, $^{b}${\tt\small shaunak@egr.msu.edu}, $^{c}${\tt\small dshishik@gmu.edu}, $^{d}${\tt\small dmaity@charlotte.edu}, $^{e}${\tt\small michael.r.dorothy.civ@army.mil.}}
\thanks{$^{1}$Michigan State University, East Lansing, Michigan, $^{2}$George Meson University, Fairfax, Virginia, $^{3}$University of North Carolina,  Charlotte, North Carolina, $^{4}$DEVCOM Army Research Laboratory. }
}
\begin{document}

\maketitle
\thispagestyle{empty}
\pagestyle{empty}

\begin{abstract}
We consider a hide-and-seek game between a Hider and a Seeker over a finite set of locations. The Hider chooses one location to conceal a stationary treasure, while the Seeker visits the locations sequentially along a route.
As the search progresses, the Hider observes a prefix of the Seeker's route. After observing this information, the Hider has the option to relocate the treasure at most once to another unvisited location by paying a switching cost.

We study two seeker models. 
{In the first, the Seeker is unaware of the fact that the Hider can relocate. In the second, the Seeker select its route while accounting for the possibility that the Hider observes its path and reallocates.}
For the restricted case, we define the value-of-information created by the reveal and derive {upper bounds in terms of the switching cost using a worst-case evaluation over routes.}
{We also show that seeker awareness reduces the game value, with the difference between the restricted and feedback models bounded by the entry-wise gap between the corresponding payoff matrices.} 
Numerical examples show how this benefit decreases as the switching cost increases and as the reveal occurs later along the route.
\end{abstract}

\section{INTRODUCTION}

Hide-and-seek search problems commonly arise in several diverse applications such as security,
surveillance, and autonomous exploration. In the classical setting,
the Hider chooses a location while the Seeker chooses a route that visits
the candidate locations in a specific/admissible manner. The final outcome in terms of the time it takes for the search or the total distance covered is determined by this
pair of choices. In many settings, however, the search unfolds sequentially, and partial route information about the route taken by the Seeker may become available to the Hider before the search is completed. {This paper is about the utility of such information in terms of the cost of switching the location for the Hider.}

\subsection{Related work}

Hide-and-seek and search games have been studied extensively in the operations research and game theory literature. These works can be broadly grouped into discrete hide-and-seek games over finite locations, geometric search games in continuous environments, and dynamic search models that incorporate information updates during the search.

Hide-and-seek search problems over discrete locations were studied in early work such as \cite{norris1962studies,bram19632,neuts,efron1964optimum}. In these models, the Hider chooses a location for a stationary target,
and the Searcher inspects locations one by one according to a chosen
order. The search cost is measured by the time or travel distance
required to locate the target. The hiding location and the inspection order together determine the
result of the search. This leads to a zero-sum game defined over the
possible search sequences. The works
\cite{roberts1978search, gittins1979search} studied optimal search strategies in this
setting. Broader treatments of search games and their connections to search theory appear in \cite{alpern2003theory,stone1976theory}. Connections between search problems and pursuit–evasion dynamics have also been studied in the differential game framework introduced in \cite{isaacs1999differential}. 

A related body of work considers geometric search games in which the Searcher moves in a continuous or network environment and incurs travel cost until detection. Linear search problems studied in \cite{beck1970yet} and later in \cite{gal1972general,gal1974discrete} consider
settings where the Searcher moves in a continuous environment and the
travel distance until detection determines the cost. In these formulations, the hiding location is fixed once it is chosen.

When the number of possible search routes becomes large, solving the resulting zero-sum matrix games can become computationally challenging.  The study in \cite{bopardikar2013randomized}  presented randomized sampling methods for computing approximate security policies in large matrix games. Their approach constructs smaller sampled subgames that provide high-confidence guarantees on the security value, and is demonstrated on a hide-and-seek search problem with exponentially many search routes.

In the formulations cited above, the Hider chooses a location, and the Searcher selects a search plan before the search begins. Once these choices are made, the payoff depends only on the resulting route and the hiding location. In many practical search settings, however, the search unfolds over time, and partial information about the search trajectory may become available before the search is completed. Dynamic aspects of search and information acquisition have been studied in related settings such as search allocation and search games with information updates \cite{hohzaki2016search}.

{Recent work has also used value-of-information objectives in deceptive path planning, where an agent selects low-information trajectories to induce suboptimal adversarial interventions\cite{suttle2025value}.} In \cite{rostobaya2025deceptive}, authors study a deceptive path–planning problem in which an agent moves toward its goal while attempting to mislead an observer about the destination.
The observer updates its belief based on the observed trajectory and allocates defensive resources accordingly.
These observations motivate the model studied here.

In our proposed model, a part of the Seeker's route may become visible during the search. After observing this prefix, the Hider may
move the treasure once to another unvisited location by paying a switching cost. The relocation occurs after the reveal stage and creates a second decision point for the Hider. We analyze this model through the notion of value-of-information (VOI) and study how partial route revelation affects equilibrium behavior. 

\subsection{Contributions}

This paper studies a two-stage hide-and-seek game with partial route revelation. After the equilibrium of the game without relocation is determined, 
the Hider observes a prefix of the Seeker's route at a prescribed reveal time. The Hider may then relocate once to another unvisited location by paying a switching cost.

We formally quantify the benefit of this partial {route revelation} using an appropriately defined value-of-information (VOI) and {evaluate its dependence on} the switching cost and the reveal time. We compare two settings: one in which the Seeker is unaware of the possibility of switching and one in which the Seeker anticipates it. This comparison clarifies how strategic awareness affects the equilibrium value and the players' strategies.
The main contributions of this work are summarized below.
{\begin{itemize}
\item \textbf{Structural results:}
We identify conditions under which switching yields no improvement for the Hider. In particular, if the base payoff matrix admits a unique pure saddle point, switching provides no incentive. We also derive the probability that the game terminates before $t_{\mathrm{reveal}}$ under mixed strategies.
\item \textbf{Characterization of the VOI:}
We construct a reveal-stage payoff representation that incorporates the information set and switching cost. This representation enables the definition and computation 
of route-level, location-level, and expected value-of-information (VOI) under the restricted seeker model.
\item \textbf{Parametric analysis:}
We {characterize the relationship} between VOI, the switching cost $c$, and the reveal time $t_{\mathrm{reveal}}$, and identify parameter regimes in which partial information benefits the Hider.
\item \textbf{Impact of strategic awareness:}
We compare the restricted and seeker-aware models and show how anticipation of switching modifies the equilibrium value and strategies. 
\end{itemize}
}
We now formalize the hide-and-seek model described above and introduce the notation used throughout the paper.

\section{Problem Formulation}
We consider a two–player hide-and-seek game between a
\emph{Hider} and a \emph{Seeker}.
The Hider (maximizer) selects a hiding location for a stationary treasure, while the Seeker (minimizer)
selects a route along which the locations are visited sequentially. Depending on the model, the Seeker may either commit to a route without anticipating relocation (restricted seeker) or adapt {subsequent choice of route} at a designated reveal time (seeker-aware model). 
Let
$$
\mathcal{N} = \{1,2,\dots,N\}
$$
denote a finite set of candidate hiding locations.
The Seeker chooses a route from the set
$$
\mathcal{R} = \{r_1, r_2, \dots, r_M\}, \quad M = N!,
$$
where each route $r_j$ is a permutation of $\mathcal{N}$, such as
$$
r_j = (r_{j,1}, r_{j,2}, \dots, r_{j,N}).
$$

\smallskip

The Seeker starts from a fixed origin $O$. If $(u, v) \in \mathcal{N} \times \mathcal{N}$, we denote the distance from the origin to location $u$
by $d_{0u} \ge 0$, while the pairwise distance between two points $u$ and $v$ is denoted by $d_{uv} \ge 0$.
If the Hider selects location $i \in \mathcal{N}$ and {chooses to stay there for the rest of the game}, while the Seeker chooses route $r_j$, then the payoff is the total distance traveled to reach location $i$ along route $r_j$. We consider this setting as our \emph{baseline} setting. Therefore, we define the baseline payoff matrix
$
A \in \mathbb{R}^{M \times N},
$
where
\begin{align}\label{eq:baseA}
A(j,i)
=
d_{0,r_{j,1}}
+
\sum_{k=1}^{\tau(j,i)-1}
d_{r_{j,k},\, r_{j,k+1}},
\end{align}
and $\tau(j,i)$ denotes the position of location $i$ in route $r_j$.
Thus, $A(j, \, i)$ represents the cumulative travel distance required to reach location $i$ when route $r_j$ is followed.

\smallskip

The game unfolds in two stages.

\smallskip

\emph{Stage 0:}  
The Hider selects an initial hiding location $i \in \mathcal{N}$ and the Seeker selects a route 
$r_j \in \mathcal{R}$. 

These selections are determined by a mixed saddle-point equilibrium of the baseline zero–sum game with payoff matrix $A$. In particular, the Hider and the Seeker compute mixed strategies
$$
z^\star \in \Delta_N,
\qquad
y^\star \in \Delta_M,
$$
that solves the minimax problem
\begin{equation}
v^{\mathrm{base}}
=
\max_{z \in \Delta_N}
\min_{y \in \Delta_M}
y^\top A z.
\label{eq:vbase}
\end{equation}
The probability simplices are defined as
$$
\Delta_N
=
\left\{
z \in \mathbb{R}^N :
z_i \ge 0, \;
\sum_{i=1}^N z_i = 1
\right\},
$$
$$
\Delta_M
=
\left\{
y \in \mathbb{R}^M :
y_j \ge 0, \;
\sum_{j=1}^M y_j = 1
\right\}.
$$
The initial hiding location and route are then drawn according to
$z^\star$ and $y^\star$, respectively. 
{\begin{assumption}\label{as:1}
    The Hider knows only the support of the selected
baseline strategy $y^\star$, but not the realized route.
\end{assumption}}

\smallskip

\emph{Reveal Stage:}  
After the Seeker visits the first $t_{\mathrm{reveal}} \in \{1,\cdots N-1\}$ locations of the chosen route, the Hider observes the visited prefix
$
h = (r_{j,1}, \dots, r_{j,t_{\mathrm{reveal}}}).
$
The set of observed prefixes induces the information set
$
\mathcal{I}(h)
=
\left\{
k \in \{1,\dots,M\} :
(r_{k,1},\dots,r_{k,t_{\mathrm{reveal}}}) = h
\right\}.
$
Now, we define the visited and unvisited sets of locations as
$$
\mathcal{V}(h) = \{r_{j,1},\dots,r_{j,t_{\mathrm{reveal}}}\},
\qquad
\mathcal{U}(h) = \mathcal{N} \setminus \mathcal{V}(h).
$$

If $i \notin \mathcal{V}(h)$, the hider takes one of the following actions:
\begin{itemize}
\item \textbf{Stay:} keep the treasure at the original location $i$.
\item \textbf{Switch:} relocate once to some $\hat{i} \in \mathcal{U}(h)\setminus\{i\}$ by paying a switching cost $c (d_{i \hat{i}}) \ge 0$.
\end{itemize}
The switch may occur at most once and is irreversible. Moreover, to ensure that the game terminates, switching to any location in $\mathcal{V}(h)$ is not allowed.
In the remainder of this paper, for the ease of exposition, we assume a constant switching cost. Therefore, the switching cost is denoted by $c \ge 0$ for every admissible switch $(i,\hat{i})$.

\medskip

The evaluation of the game after the reveal stage depends on the Seeker's ability to adapt. We consider two seeker models:

\begin{enumerate}
    \item \textbf{Restricted Seeker:} In this model, the Seeker commits to a route at Stage 0 and cannot modify it during the game. This can also be thought of as a seeker who is unaware of the Hider's ability to switch.
    
    \item \textbf{Feedback (Seeker-Aware) Model:}
    In this model, the Seeker anticipates the Hider's switching capability 
    {
and the players solve the corresponding zero-sum sub-game
over the prefix-consistent routes and admissible stay/switch actions.}
\end{enumerate}

The objective of this work is to understand how partial route revelation and a one-time switching option affect the equilibrium of the hide-and-seek game. We quantify the benefit obtained by the Hider from observing the visited prefix at time $t_{\mathrm{reveal}}$ and define the corresponding value-of-information (VOI). We then study how this gain depends on the switching cost $c$ and $t_{\mathrm{reveal}}$. Finally, we compare the restricted and seeker-aware models to determine how strategic anticipation of switching alters the equilibrium value and strategies.

We now present preliminary results that clarify the baseline game's behavior before incorporating post-reveal adaptation in the next section.
\medskip

\section{Preliminary Results}
In this section, we examine the structural properties of the base game defined by the payoff matrix $A$. {The baseline hide-and-seek interaction defined by the payoff matrix $A$ is a finite zero–sum matrix game. Such games always admit equilibrium
in mixed strategies, and in general, the optimal strategies of both
players are not pure. When the Seeker randomizes over several routes,
the Hider does not know in advance which trajectory will be followed.
Observing a prefix of the route, therefore, reveals partial information
about the realized strategy and may eliminate some routes from the
support of the Seeker's mixed strategy. This updated information can
create an incentive for the Hider to change the hiding location after
the reveal stage.}

{The next lemma describes a special case in which the baseline game has a unique pure saddle point. In this case, the Seeker's route is known in advance, and the reveal stage does not affect the Hider's decision.}




\begin{lemma}[Unique pure saddle point case]
    Suppose the baseline game admits a unique pure saddle point $({r_j^\star},i^\star)$, that is,
\begin{align}
A(r_j^\star,i^\star)
=
\max_{i \in \mathcal{N}} A(r_j^\star,i)
=
\min_{r_j \in \mathcal{R}} A(r_j,i^\star).
\end{align}
Then, even if the Hider observes partial route information at time $t_{\mathrm{reveal}}$
and is allowed to switch once to an unvisited location at a cost $c \ge 0$,
there is no incentive to switch. Staying at $i^\star$ remains optimal.
\end{lemma}

\begin{proof}
Since $(r_j^\star,i^\star)$ is a pure saddle point,
$i^\star$ is {the unique} best response to $r_j^\star$, and hence
\begin{align}
  A(r_j^\star,i^\star) \ge A(r_j^\star,i),
\quad \forall i \in \mathcal{N}.  
\end{align}
If the hider switches to some admissible location $\hat{i} \neq i^\star$ by paying $c$, then
the resulting payoff becomes
\begin{align}
   A(r_j^\star,\hat{i}) - c. 
\end{align}
Because $c \ge 0$ and
\begin{align}
   A(r_j^\star,i^\star) \ge A(r_j^\star,\hat{i}),
\end{align}
we obtain
\begin{align}
    A(r_j^\star,\hat{i}) - c
\le
A(r_j^\star,i^\star).
\end{align}
Thus, switching cannot increase the Hider's payoff. Therefore, staying at $i^\star$ is optimal for the Hider, as it maximizes the Hider's payoff.
\end{proof}

\smallskip

It is worth noting that the existence of a pure saddle point for the
zero-sum game defined by $A$ in Eq. \eqref{eq:baseA} is rare. {Such a case may occur, for example, when all candidate locations together with the origin are collinear. In such a geometry, the ordering of visits is essentially fixed, and the induced payoff matrix can admit a unique pure equilibrium. In the remainder of the paper, we therefore focus on the general case in which the equilibrium requires mixed strategies.}

 \smallskip

The following lemma characterizes the probability that the treasure is found by time $t_{\mathrm{reveal}}$ under mixed strategies.

 \smallskip

\begin{lemma}[Probability of termination by reveal time]
Consider a pair of mixed strategies $y \in \Delta_M$ and $z \in \Delta_N$.
Let $f_k(r_j)$ denote the location visited at time $k$ along route $r_j$,
i.e., $f_k(r_j)=r_{j,k}$.
Then, for any $t_{\mathrm{reveal}} \in \{1,\dots,N\}$, the probability that the game
terminates by time $t_{\mathrm{reveal}}$ is
\begin{equation}
\mathbb P(\text{end by } t_{\mathrm{reveal}})
=
\sum_{i=1}^N z_i
\left(
\sum_{k=1}^{t_{\mathrm{reveal}}}
\sum_{j:\, f_k(r_j)=i} y_j
\right).
\label{eq:endprob_reveal}
\end{equation}
\end{lemma}

\begin{proof}
We prove Eq.~\eqref{eq:endprob_reveal} by induction on $t_{\mathrm{reveal}}$.

\textbf{Base case ($t_{\mathrm{reveal}}=1$).}
The game terminates by time $1$ if and only if the Hider selects location $i$
and the Seeker visits $i$ at time $1$, i.e., $f_1(r_j)=i$.
Therefore,
\begin{align}
\mathbb P(\text{end by }1)
&=
\sum_{i=1}^N
z_i
\sum_{j:\, f_1(r_j)=i} y_j,
\end{align}
which matches Eq.~\eqref{eq:endprob_reveal} for $t_{\mathrm{reveal}}=1$.

\textbf{Induction step.}
Assume Eq.~\eqref{eq:endprob_reveal} holds for $t_{\mathrm{reveal}}=t$,
where $1 \le t < N$.
Then
\begin{align}
\mathbb P(\text{end by }t+1)
&=
\mathbb P(\text{end by }t)
+
\mathbb P(\text{end exactly at }t+1).
\label{eq:ind_split}
\end{align}

The event \emph{``end exactly at $t+1$"} occurs when the Hider selects
location $i$ and the Seeker visits $i$ at time $t+1$, that is,
$f_{t+1}(r_j)=i$.
Hence,
\begin{align}
\mathbb P(\text{end exactly at }t+1)
=
\sum_{i=1}^N z_i
\sum_{j:\, f_{t+1}(r_j)=i} y_j.
\label{eq:end_exact}
\end{align}

Substituting the induction hypothesis into Eq.~\eqref{eq:ind_split}
and using Eq.~\eqref{eq:end_exact}, we obtain
\begin{align}
\mathbb P(\text{end by }t+1)
&=
\sum_{i=1}^N z_i
\left(
\sum_{k=1}^{t}
\sum_{j:\, f_k(r_j)=i} y_j
\right)
\\
&\quad+
\sum_{i=1}^N z_i
\sum_{j:\, f_{t+1}(r_j)=i} y_j
\nonumber\\
&=
\sum_{i=1}^N z_i
\left(
\sum_{k=1}^{t+1}
\sum_{j:\, f_k(r_j)=i} y_j
\right),
\end{align}
which is Eq.~\eqref{eq:endprob_reveal} with $t_{\mathrm{reveal}}=t+1$.
Therefore the result holds for all
$t_{\mathrm{reveal}}\in\{1,\dots,N\}$.
\end{proof}

\medskip

\section{Main Results}
We now present the main results for the game with partial route revelation. The analysis is divided into two parts. First, we study the restricted-seeker model and characterize the {value-of-information} created by the reveal. Next, we consider the seeker-aware model and establish bounds that relate the corresponding game values.
\subsection{Analysis with restricted seeker model}
In the restricted seeker model, the Seeker commits to a route at Stage~0
using only a priori information.
The Seeker solves the baseline matrix game defined by the payoff matrix $A$
and selects a mixed strategy $y^\star \in \Delta_M$. Since there can be infinitely many saddle point strategies, the strategy $y^\star$ is not known to the Hider. 
Instead, the Hider can only infer the supports of the mixed strategy $y^\star$ as per Assumption \ref{as:1}.
The Hider selects a hiding location $i \in \mathcal N$ according to $z^\star \in \Delta_N$ at $t=0$. Figure \ref{fig:reseek} represents game tree of restricted seeker model with $N=2$.
\begin{figure}[!h]
    \centering
    \includegraphics[width=\linewidth]{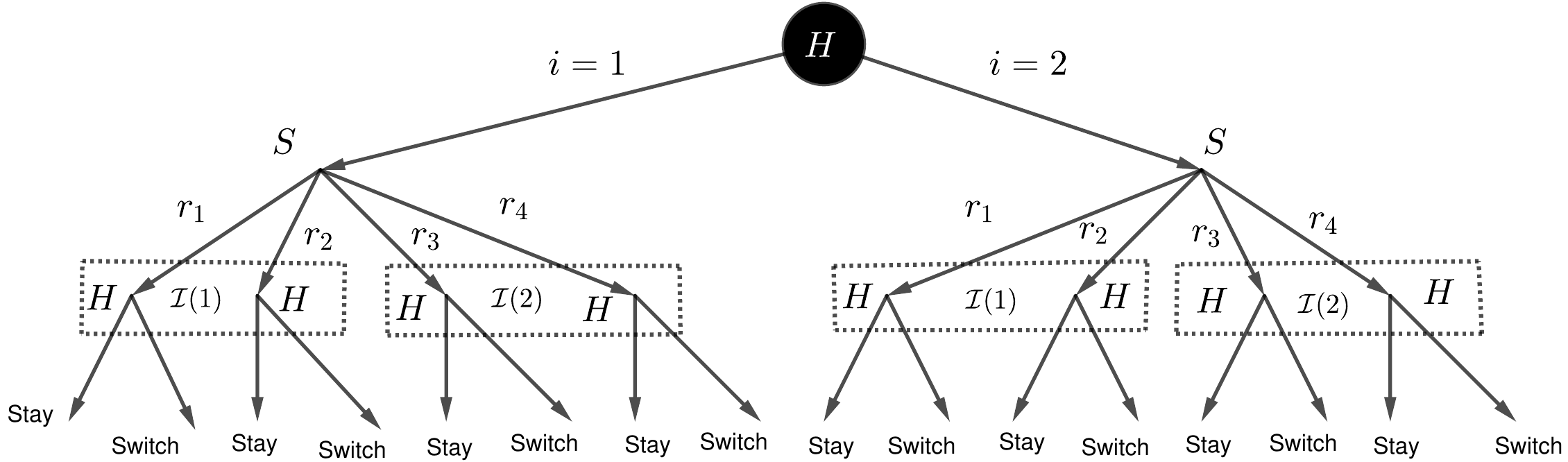}
    \caption{Game Tree for Restricted Seeker Model for $N=2$.}
    \label{fig:reseek}
\end{figure} 

 \begin{algorithm}[!h]
\caption{Building $A^{\mathrm{switch}}$ (restricted-seeker model)}
\label{alg:Aswitch}
\begin{algorithmic}[1]
\Require $A \in \mathbb{R}^{M \times N}$, routes $\{r_j\}_{j=1}^M$, 
         reveal time $t_{\mathrm{reveal}}$, switching cost $c \ge 0$
\Ensure $A^{\mathrm{switch}} \in \mathbb{R}^{M \times N}$

\For{$j = 1,\dots,M$}
    \State $h_j = (r_{j,1},\dots,r_{j,t_{\mathrm{reveal}}})$
    \State $\mathcal{V}(h_j) = \{r_{j,1},\dots,r_{j,t_{\mathrm{reveal}}}\}$
    \State $\mathcal{U}(h_j) = \mathcal{N} \setminus \mathcal{V}(h_j)$

    \For{$i = 1,\dots,N$}
        \If{$i \in \mathcal{V}(h_j)$}
            \State $A^{\mathrm{switch}}(j,i) = A(j,i)$
        \Else
            \ForAll{$\hat{i} \in \mathcal{U}(h_j)$}
                \If{$\hat{i} = i$} \State $$A^{\mathrm{red}}(h_j,i,\hat{i},c)  = A(j,\hat{i}) - A(j,r_{j,t_{\mathrm{reveal}}})$$
                \Else \State $$A^{\mathrm{red}}(h_j,i,\hat{i},c) 
                    = A(j,\hat{i}) - A(j,r_{j,t_{\mathrm{reveal}}}) - c$$
                \EndIf
            \EndFor
           \State
           {\begin{align*}
               A^{\mathrm{switch}}(j,i)
&=
A\big(j,r_{j,t_{\mathrm{reveal}}}\big)
+
\max_{\hat{i}\in\mathcal{U}(h_j)}
A^{\mathrm{red}}(h_j,i,\hat{i},c)
           \end{align*}}
        \EndIf
    \EndFor
\EndFor
\end{algorithmic}
\end{algorithm}
At time $t_{\mathrm{reveal}}$, the Hider observes the visited set
$\mathcal V(h_j)$ generated by the realized route $r_j$.
If $i \in \mathcal U(h_j)$, the Hider may either stay at $i$
or switch once to some
$\hat i \in \mathcal U(h_j)\setminus\{i\}$
by paying a cost of $c$.


Formally, fix a realized route $r_j=(r_{j,1},\dots,r_{j,N})$ and a reveal time $t_{\mathrm{reveal}}$.
Let $h_j=(r_{j,1},\dots,r_{j,t_{\mathrm{reveal}}})$ denote the revealed prefix and
let $r_{j,t_{\mathrm{reveal}}}$ be the node visited at the reveal time.
For an initial hiding location $i\in\mathcal U(h_j)$ and any admissible post-reveal
location $\hat i\in\mathcal U(h_j)$, define
\begin{align}
A^{\mathrm{red}}(h_j,i,\hat i,c)
=
\begin{cases}
A(j,\hat i)-A\big(j,r_{j,t_{\mathrm{reveal}}}\big),
& \text{if } \hat i=i,\\[4pt]
A(j,\hat i)-A\big(j,r_{j,t_{\mathrm{reveal}}}\big)-c,
&\text{if } \hat i\neq i,
\end{cases}
\label{eq:Ared_restricted}
\end{align}
where $A(j,\hat i)-A\big(j,r_{j,t_{\mathrm{reveal}}}\big)$ is the remaining travel
distance along the committed route $r_j$ from the reveal node to $\hat i$.
The quantity $A^{\mathrm{red}}(h_j,i,\hat i,c)$ denotes the reveal-stage payoff
under route $r_j$. If the Hider stays at $i$ ($\hat{i} = i$), the payoff equals the remaining travel distance along $r_j$ to $i$. If the Hider switches to
$\hat i\neq i$, the payoff equals the remaining distance to $\hat i$ minus
the switching cost $c$. Thus, Eq.~\eqref{eq:Ared_restricted} captures the
post-reveal stay/switch decision along the committed route.

In reveal stage, the Hider will then choose the location
$\hat i \in \mathcal U(h_j)$
that maximizes $A^{\mathrm{red}}(h_j,i,\hat i,c)$.
Thus, the route-wise switching payoff is
{\begin{equation}
A^{\mathrm{switch}}(j,i)
=
A\big(j,r_{j,t_{\mathrm{reveal}}}\big)
+
\max_{\hat i\in\mathcal U(h_j)}
A^{\mathrm{red}}(h_j,i,\hat i,c).
\label{eq:Aswitch_restricted}
\end{equation}}

{For locations that have already been visited at the reveal time,
{there is no option to switch}. If $i \in \mathcal V(h_j)$,
the game is already ended 
before the reveal stage,
and the payoff equals the baseline cost. In this case, 
$
A^{\mathrm{switch}}(j,i) = A(j,i)$.
{Thus $A^{\mathrm{switch}}(j,i)$ is defined for all
$j \in \{1,\dots,M\}$ and $i \in \mathcal N$, and the resulting
matrix $A^{\mathrm{switch}} \in \mathbb{R}^{M \times N}$ has the
same dimensions as the baseline payoff matrix $A$. Note that, the matrix $A^{\mathrm{switch}}$ depends on the switching cost $c$ through
Eq.~\eqref{eq:Ared_restricted}. For notational simplicity, this dependence is not
made explicit unless needed.}  

Algorithm~\ref{alg:Aswitch} {summarizes these steps involved in constructing} the effective payoff matrix
$A^{\mathrm{switch}}$ by embedding the reveal-stage stay/switch decision into each route-location pair. By construction,
$
A^{\mathrm{switch}} \in \mathbb{R}^{M \times N},
$
with rows indexed by routes and columns indexed by hiding locations. 

{Although the restricted Seeker remains committed to the baseline
strategy and does not adapt after the reveal, $A^{\mathrm{switch}}$
captures the Hider's optimal stay/switch payoff for each
route - location pair. For comparison, we evaluate its saddle-point
value against the Seeker's best response. The corresponding benchmark
value is}
\begin{equation}
v^{\mathrm{switch}}
=
\max_{z \in \Delta_N}
\min_{y \in \Delta_M}
y^\top A^{\mathrm{switch}} z.
\label{eq:vswitch_restricted}
\end{equation}

{For each route - location pair, the Hider always has the option to keep the original hiding location. Therefore, the switching model cannot produce a smaller payoff for the Hider than the baseline game, and the game value satisfies
\begin{align}\label{eq:vswitch_restricted}
v^{\mathrm{switch}} \ge v^{\mathrm{base}}.
\end{align}}



\medskip




\begin{algorithm}[H]
\caption{Hider's Strategy in Restricted Seeker Model}
\label{alg:hider_strategy}
\begin{algorithmic}[1]
\Statex \textbf{Stage 0 (Hiding Location Selection):}
\State The Hider knows the supports of the Seeker's mixed strategy ${y}^\star$.
\State The Hider selects a hiding location $i$ from the set of locations $\mathcal{N}$ according to ${z}^\star \in \Delta_N$.

\Statex \textbf{Reveal Stage:}
\State At the reveal time $t_{\mathrm{reveal}}$, the Hider observes the visited set $\mathcal{V}({h}_j)$ generated by the realized Seeker route ${r}_j$.
\If{$i \in \mathcal{U}({h}_j)$}
    \State The Hider has two options:
    \State \hspace{2em} a. Stay at the initial hiding location $i$.
    \State \hspace{2em} b. Switch to a different location $\hat{i}$ in the unvisited set $\mathcal{U}({h}_j)$, incurring a switching cost $c$.
    \State The Hider chooses the option that maximizes the reduced payoff $A^{\mathrm{red}}({h}_j, i, \hat{i}, c)$, as defined in Equation~\eqref{eq:Ared_restricted}.
\EndIf
\State The effective payoff matrix $\mathbf{A}^{\mathrm{switch}}$ is constructed according to Algorithm~\ref{alg:Aswitch}.
\end{algorithmic}
\end{algorithm}

Algorithm~\ref{alg:hider_strategy} summarizes the Hider's decision process
in the restricted-seeker model.
The reveal stage introduces an additional strategic option for the Hider. Upon observing the revealed prefix, the Hider may relocate to another admissible location before the Seeker completes the route. This relocation can alter the payoff relative to remaining at the original location. 

{
Since the Hider knows the support of the Seeker's baseline strategy
$y^\star$, but does not know the realized continuation beyond the
observed prefix, define the set of supported routes consistent with
$h_j$ as
\begin{equation}
\mathcal I^\star(h_j)
=
\left\{
k\in\mathcal I(h_j):y_k^\star>0
\right\}.
\label{eq:supported_information_set}
\end{equation}
For $r_k\in\mathcal I^\star(h_j)$, let
$A_k^{\mathrm{red}}(h_j,i,\hat{i},c)$ denote the reduced payoff in
\eqref{eq:Ared_restricted} evaluated along route $r_k$.}

To measure this incremental benefit, we formally define the value of
information (VOI) as follows.
\smallskip
{\begin{definition}\label{def:voi}[Value-of-Information]
Let $r_j\in\operatorname{supp}(y^\star)$ be a Seeker route with
revealed prefix $h_j$, and let $i\in\mathcal U(h_j)$ be the initial
hiding location.
The route-level value-of-information is defined as the improvement
obtained by allowing the Hider to optimally use the revealed
information, relative to retaining the original hiding location:
\begin{equation}
\begin{aligned}
\mathrm{VOI}(r_j,i;c)
={}&
\max_{\hat{i}\in\mathcal U(h_j)}
\min_{k\in\mathcal I^\star(h_j)}
A_k^{\mathrm{red}}(h_j,i,\hat{i},c)
\\
&-
\min_{k\in\mathcal I^\star(h_j)}
A_k^{\mathrm{red}}(h_j,i,i,c).
\end{aligned}
\label{eq:voi_restricted}
\end{equation}
$\hfill \blacksquare$
\end{definition} 
The first term in \eqref{eq:voi_restricted} is the best
reveal-stage payoff available to the Hider, evaluated in the worst
case over all supported Seeker routes consistent with the observed
prefix. The second term is the corresponding payoff when the Hider
retains its original hiding location $i$.
Since the stay action $\hat{i}=i$ is included in the maximization,
\begin{equation}
\mathrm{VOI}(r_j,i;c)\geq 0.
\end{equation}
If $i\in\mathcal V(h_j)$, the treasure has already been found by the
reveal time, and we define
\begin{equation}
\mathrm{VOI}(r_j,i;c)=0.
\end{equation}
For each hiding location $i$, define the worst-case
value-of-information as
\begin{equation}
\overline{\mathrm{VOI}}(i;c)
=
\min_{j\in\operatorname{supp}(y^\star)}
\mathrm{VOI}(r_j,i;c).
\label{eq:barvoi_restricted}
\end{equation}
The expected value-of-information under the Hider's mixed strategy is
\begin{equation}
\mathbb E[\mathrm{VOI}](c)
=
{z^\star}^\top\overline{\mathrm{VOI}}(c),
\label{eq:expected_voi}
\end{equation}
where
$
\overline{\mathrm{VOI}}(c)
=
\big[
\overline{\mathrm{VOI}}(1;c),\ldots,
\overline{\mathrm{VOI}}(N;c)
\big]^\top.
$
}
\medskip



\subsection{Critical Switching Cost}
For each route $r_j$ with revealed prefix $h_j$
and each $i\in\mathcal U(h_j)$, define the route-wise
free-switching advantage
\begin{equation}
c^\star(j,i)
=
\max_{\hat{i}\in\mathcal U(h_j)}
\Big(
A_j^{\mathrm{red}}(h_j,i,\hat{i},0)
-
A_j^{\mathrm{red}}(h_j,i,i,0)
\Big),
\label{eq:cstar_restricted}
\end{equation}
{where the subscript $j$ emphasizes that the reduced payoff is
evaluated along the fixed route $r_j$.}
Using Eq.~\eqref{eq:Ared_restricted} with $c=0$,
Eq. \eqref{eq:cstar_restricted} can be written as
\begin{align}
c^\star(j,i)
&=
\max_{\hat{i}\in\mathcal U(h_j)}
\big(A(j,\hat{i})-A(j,i)\big)
\nonumber\\
&=
A(j,r_{j,N})-A(j,i),
\label{eq:cstar_simplified}
\end{align}
since the cumulative payoff $A(j,\hat{i})$ is maximized over the
unvisited set at the last location visited along the committed route
$r_j$.

The quantity $c^\star(j,i)$ therefore represents the maximum
zero-cost advantage that switching can provide over staying at
location $i$ when the Seeker is known to follow route $r_j$.
Define
\begin{equation}
c^\star_{\mathrm{global}}
=
\max_{j \in \rm{supp}(y^\star)}
\max_{i\in\mathcal U(h_j)}
c^\star(j,i),
\label{eq:cstar_global_restricted}
\end{equation}
{which essentially upper bounds the zero-cost
switching advantage along every admissible supported route.
In the next result, we use $c^\star_{\mathrm{global}}$ to bound
$\mathbb E[\mathrm{VOI}]$ as a function of the switching cost $c$.}

\medskip

\begin{theorem}[Upper bound on $\mathbb{E}(\mathrm{VOI})$ as a function of $c$]
\label{lem:expvoi_bound}
In the restricted-seeker model, the expected value-of-information satisfies
\begin{align}
0
\le
\mathbb{E}[\mathrm{VOI}](c)
\le
\max\!\left\{
c^\star_{\mathrm{global}}-c,\,0
\right\},
\label{eq:expvoi_bound}
\end{align}
where $c^\star_{\mathrm{global}}$ is defined in
Eq.~\eqref{eq:cstar_global_restricted}.
Consequently, if $c\ge c^\star_{\mathrm{global}}$, then
$\mathbb{E}[\mathrm{VOI}](c)=0$.
\end{theorem}
{
\begin{proof}
Fix a revealed prefix $h_j$ and an initial hiding location
$i\in\mathcal U(h_j)$. For any switching action $\hat i\neq i$,
the constant switching cost gives
\begin{align}
A_k^{\mathrm{red}}(h_j,i,\hat i,c)
&=
A_k^{\mathrm{red}}(h_j,i,\hat i,0)-c,\\
A_k^{\mathrm{red}}(h_j,i,i,c)
&=
A_k^{\mathrm{red}}(h_j,i,i,0).
\end{align}
Therefore,
\begin{align}
&\min_{k\in\mathcal I^\star(h_j)}
A_k^{\mathrm{red}}(h_j,i,\hat i,c)
-
\min_{k\in\mathcal I^\star(h_j)}
A_k^{\mathrm{red}}(h_j,i,i,c)
\nonumber\\
&=
\min_{k\in\mathcal I^\star(h_j)}
A_k^{\mathrm{red}}(h_j,i,\hat i,0)
-
\min_{k\in\mathcal I^\star(h_j)}
A_k^{\mathrm{red}}(h_j,i,i,0)
-c.
\label{eq:voi_pf_step1}
\end{align}
Using
\begin{equation}
\min_k a_k-\min_k b_k
\le
\max_k(a_k-b_k),
\end{equation}
which follows by choosing
$k_b\in\arg\min_k b_k$, so that
\[
\min_k a_k-\min_k b_k
\le
a_{k_b}-b_{k_b}
\le
\max_k(a_k-b_k),
\]
Eq.~\eqref{eq:voi_pf_step1} gives
\begin{align}
&\min_{k\in\mathcal I^\star(h_j)}
A_k^{\mathrm{red}}(h_j,i,\hat i,c)
-
\min_{k\in\mathcal I^\star(h_j)}
A_k^{\mathrm{red}}(h_j,i,i,c)
\nonumber\\
&\le
\max_{k\in\mathcal I^\star(h_j)}
\Big(
A_k^{\mathrm{red}}(h_j,i,\hat i,0)
-
A_k^{\mathrm{red}}(h_j,i,i,0)
\Big)-c.
\label{eq:voi_pf_step2}
\end{align}
From the definition of the route-wise critical cost,
\begin{align}
A_k^{\mathrm{red}}(h_j,i,\hat i,0)
-
A_k^{\mathrm{red}}(h_j,i,i,0)
\le
c^\star(k,i)
\le
c^\star_{\mathrm{global}}.
\end{align}
Hence, from Eq. \eqref{eq:cstar_global_restricted}, for every switching action $\hat i\neq i$,
\begin{align}
&\min_{k\in\mathcal I^\star(h_j)}
A_k^{\mathrm{red}}(h_j,i,\hat i,c)
-
\min_{k\in\mathcal I^\star(h_j)}
A_k^{\mathrm{red}}(h_j,i,i,c)
\nonumber\\
&\le
c^\star_{\mathrm{global}}-c.
\end{align}
Since the stay action $\hat i=i$ is always feasible,
Definition~\ref{def:voi} gives
\begin{align}
0
\le
\mathrm{VOI}(r_j,i;c)
\le
\max\!\left\{
c^\star_{\mathrm{global}}-c,\,0
\right\}.
\end{align}
Taking the minimum over supported routes gives
\begin{align}
0
\le
\overline{\mathrm{VOI}}(i;c)
\le
\max\!\left\{
c^\star_{\mathrm{global}}-c,\,0
\right\}.
\end{align}
Finally, since
\begin{equation}
\mathbb E[\mathrm{VOI}](c)
=
{z^\star}^\top\overline{\mathrm{VOI}}(c),
\qquad z^\star \in\Delta_N,
\end{equation}
the same upper bound holds for the expected value:
\begin{align}
0
\le
\mathbb E[\mathrm{VOI}](c)
\le
\max\!\left\{
c^\star_{\mathrm{global}}-c,\,0
\right\}.
\end{align}
If $c\ge c^\star_{\mathrm{global}}$, the right-hand side is zero,
and therefore $\mathbb E[\mathrm{VOI}](c)=0$.
\end{proof}}

\smallskip

 We discuss the dependence of VOI on $t_\mathrm{reveal}$ next.
 To make the dependence of VOI on the reveal time explicit, we write
$\mathrm{VOI}(r_j,i,t_{\mathrm{reveal}})$ for the quantity in
Definition~\ref{def:voi} evaluated at the prefix $h_j(t_{\mathrm{reveal}})$.
 \smallskip

\begin{proposition}[Effect of reveal time on VOI]
\label{lem:treveal_mono}
The expected value-of-information,
viewed as a function of $t_{\mathrm{reveal}}$, is non-increasing:
\begin{equation}
\mathbb E[\mathrm{VOI}](t_{\mathrm{reveal}}{+}1)
\le
\mathbb E[\mathrm{VOI}](t_{\mathrm{reveal}}),
\label{eq:expvoi_mono_lem}
\end{equation}
for $t_{\mathrm{reveal}}\in\{1,\dots,N{-}1\}$.
\end{proposition}

\begin{proof} 
By definition,
$\mathcal U(h_j(t_{\mathrm{reveal}}))$ is the unvisited set at $t_\mathrm{reveal}$.
Since the prefix length increases by one when passing from
$t_{\mathrm{reveal}}$ to $t_{\mathrm{reveal}}{+}1$, one additional node is visited.
Therefore,
\begin{align}
\mathcal U \big(h_j(t_{\mathrm{reveal}}{+}1)\big)
&=
\mathcal U \big(h_j(t_{\mathrm{reveal}})\big)\setminus\{r_{j,t_{\mathrm{reveal}}{+}1}\},
 \\ 
\mathcal U\!\big(h_j(t_{\mathrm{reveal}}{+}1)\big)&\subseteq
\mathcal U\!\big(h_j(t_{\mathrm{reveal}})\big).
\label{eq:U_shrink_pf}
\end{align}

\smallskip

For any $i\in\mathcal U(h_j(t_{\mathrm{reveal}}))$, define the time-indexed
route-wise threshold by
\begin{align}
\begin{split}
   &c^\star(j,i,t_{\mathrm{reveal}})
=
\\&~~~~~\max_{\hat{i}\in \mathcal U(h_j(t_{\mathrm{reveal}}))}
\Big(
A^{{\mathrm{red}}}(h_j(t_{\mathrm{reveal}}),i,\hat{i},0)
\\&~~~~~~~~~~~~~~~~~~~~~~~~~~~~-
A^{{\mathrm{red}}}(h_j(t_{\mathrm{reveal}}),i,i,0)
\Big),
\label{eq:cstar_time_def} 
\end{split}
\end{align}
which matches Eq.~\eqref{eq:cstar_restricted} with explicit dependence on
$t_{\mathrm{reveal}}$ and with $c=0$.
By Eq.~\eqref{eq:U_shrink_pf}, this feasible set shrinks at
$t_{\mathrm{reveal}}{+}1$. Hence, a maximization over a smaller set cannot
exceed the maximization over a {set that contains the smaller set}, and
\begin{equation}
c^\star(j,i,t_{\mathrm{reveal}}{+}1)
\le
c^\star(j,i,t_{\mathrm{reveal}}).
\label{eq:cstar_mono_pf}
\end{equation}

Moreover, the value of $c^\star(j,i,t_{\mathrm{reveal}})$ 
can drop only if the node removed at
$t_{\mathrm{reveal}}{+}1$ is a maximizer at time
$t_{\mathrm{reveal}}$.
Indeed, letting $\hat i^\star(t_{\mathrm{reveal}})$ be any maximizer in
Eq.~\eqref{eq:cstar_time_def} at time $t_{\mathrm{reveal}}$,
Eq.~\eqref{eq:U_shrink_pf} implies:
\begin{align}
\nonumber
r_{j,t_{\mathrm{reveal}}{+}1}&\neq \hat i^\star(t_{\mathrm{reveal}})\\
\ \implies 
c^\star(j,i,t_{\mathrm{reveal}}{+}1)
&=
c^\star(j,i,t_{\mathrm{reveal}}).
\end{align}
Thus $c^\star(j,i)$ is typically piecewise constant in $t_{\mathrm{reveal}}$
and decreases only when the newly visited node is a maximizing destination.

\smallskip

Now, for $i\in\mathcal U(h_j(t_{\mathrm{reveal}}))$,  recall from
Definition~\ref{def:voi} that
$\mathrm{VOI}(r_j,i)$ involves a maximization over
$\hat i\in\mathcal U(h_j)$.
In the restricted model, $A^{\mathrm{switch}}(j,\cdot)$ is computed from
Eq.~\eqref{eq:Aswitch_restricted}, which maximizes over the admissible set
$\mathcal U(h_j)$.
By Eq.~\eqref{eq:U_shrink_pf}, the feasible set used in these 
maximizations shrinks as $t_{\mathrm{reveal}}$ increases.
Therefore, the route-level value $\mathrm{VOI}(r_j,i,t_{\mathrm{reveal}})$ cannot increase:
\begin{equation}
\mathrm{VOI}(r_j,i,t_{\mathrm{reveal}}{+}1)
\le
\mathrm{VOI}(r_j,i,t_{\mathrm{reveal}}).
\label{eq:voi_mono_pf}
\end{equation}

\smallskip

By Eq.~\eqref{eq:barvoi_restricted},
$\overline{\mathrm{VOI}}(i)=\min_j \mathrm{VOI}(r_j,i)$, and therefore, taking a minimum
over $j$ preserves the inequality in Eq.~\eqref{eq:voi_mono_pf}:
\begin{align}
\overline{\mathrm{VOI}}(i,{t_{\mathrm{reveal}}{+}1})
\le
\overline{\mathrm{VOI}}(i,{t_{\mathrm{reveal}}}).
\end{align}
Finally,
$\mathbb E[\mathrm{VOI}]=z^\top \overline{\mathrm{VOI}}$ with $z\in\Delta_N$,
so multiplying by $z_i\ge 0$ and summing over $i$ yields
Eq.~\eqref{eq:expvoi_mono_lem}.
\end{proof}

\smallskip

The results above characterize the restricted-seeker model and show how the
the value-of-information depends on the switching cost when the Seeker remains
committed to the initially selected route. We next consider a seeker-aware
model in which the Seeker may choose a continuation consistent with the
revealed prefix, and we derive bounds relating the corresponding game values.


\subsection{Analysis with seeker-aware model}
In the seeker-aware model, the Seeker anticipates the Hider's switching capability at the reveal stage. The game is therefore evaluated by backward induction. Figure \ref{fig:seekaware} depicts the game tree for seeker-aware model with $N=2$.
\begin{figure}[!h]
    \centering
    \includegraphics[width=\linewidth]{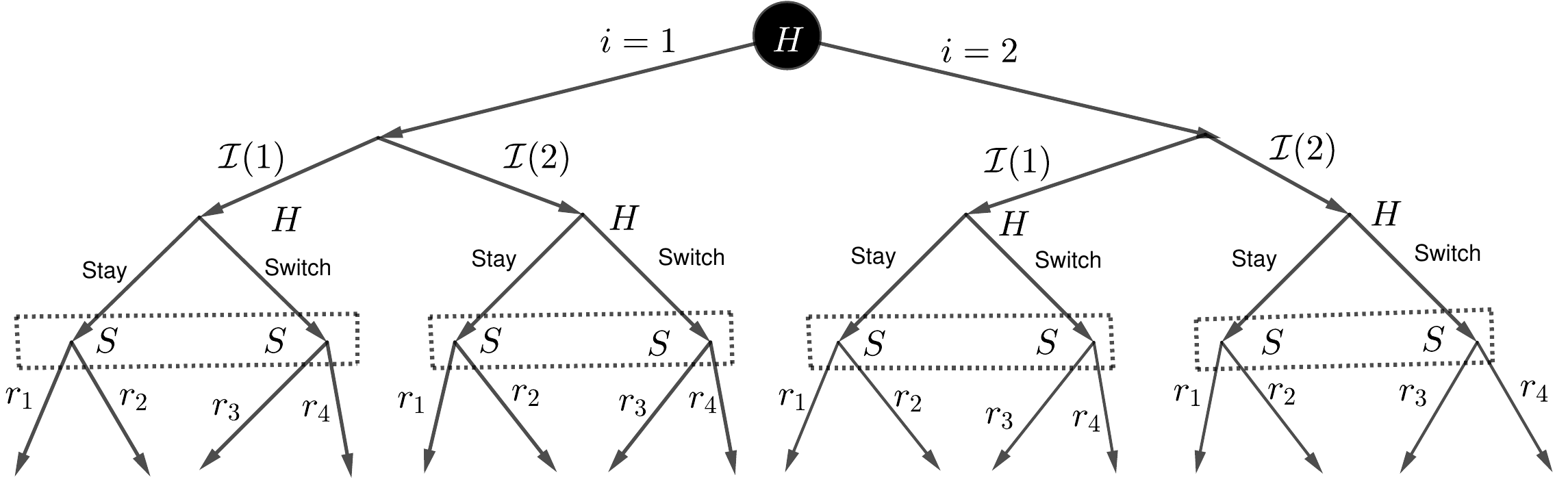}
    \caption{Game Tree for Seeker-Aware Model for $N=2$.}
    \label{fig:seekaware}
\end{figure}

Let us fix a reveal time $t_{\mathrm{reveal}}$.
For any route $r_j \in \mathcal R$, define the revealed prefix
$
h_j := (r_{j,1},\dots,r_{j,t_{\mathrm{reveal}}}).
$ The associated information set is
\begin{align}
\mathcal I(h_j)
=
\{ k \in \{1,\dots,M\} :
(r_{k,1},\dots,r_{k,t_{\mathrm{reveal}}}) = h_j \}.
\end{align}

At prefix $h_j$, if the treasure location $i$ satisfies
$i \in \mathcal V(h_j)$, the game terminates.
If $i \in \mathcal U(h_j)$, the Hider may stay at $i$
or switch once to some $\hat i \in \mathcal U(h_j) \setminus i$
by paying a cost of $c$.

\medskip

{For each prefix $h$ and initial hiding location
$i\in\mathcal U(h)$, the post-reveal interaction is a zero-sum
sub-game. The Seeker may randomize over the
prefix-consistent routes in $\mathcal I(h)$, while the Hider may
randomize over the admissible post-reveal locations in
$\mathcal U(h)$. The value of this sub-game is
\begin{align}
V(h,i)
&=
\min_{p\in\Delta_{|\mathcal I(h)|}}
\max_{q\in\Delta_{|\mathcal U(h)|}}
\sum_{\substack{k\in\mathcal I(h)\\
\hat i\in\mathcal U(h)}}
p_k q_{\hat i}
\left[
A(k,\hat i)
-
c\,\mathbf{1}_{\{\hat i\neq i\}}
\right],
\label{eq:V_feedback}
\end{align}
where $p$ and $q$ denote the mixed strategies of the Seeker and
Hider, respectively, in the sub-game.}


\medskip
\subsubsection{Feedback payoff matrix.}
Let $\mathcal{H}$ denote the set of all prefixes that arise
at time $t_{\mathrm{reveal}}$.
Define the feedback payoff matrix
$
A^{\mathrm{fb}} \in \mathbb R^{|\mathcal H| \times N}
$
entrywise as
\begin{align}
A^{\mathrm{fb}}(h,i)
=
\begin{cases}
A(j,i),
& i \in \mathcal V(h),
\\[4pt]
V(h,i),
& i \in \mathcal U(h),
\end{cases}
\label{eq:Afb_def}
\end{align}
where $r_j$ is any route generating prefix $h$.

The first case is well-defined because all routes in
$\mathcal I(h)$ coincide up to time $t_{\mathrm{reveal}}$.
Hence, if $i \in \mathcal V(h)$, the cumulative distance
$A(j,i)$ is identical for every $j$ such that $\pi(j)=h$, where
$\pi:\{1,\dots,M\}\to\mathcal H$ denotes the prefix map.
For $i\in\mathcal U(h)$, $A^{\mathrm{fb}}(h,i)$ is the value of the
corresponding zero-sum sub-game. Thus, the Seeker may
randomize over the prefix-consistent routes in $\mathcal I(h)$,
while the Hider may randomize over the admissible stay/switch actions.


\medskip
\subsubsection{Game value under feedback.}

Let $\pi:\{1,\dots,M\}\to\mathcal H$ denote the prefix map
$\pi(j)=h_j$.
A mixed strategy $y\in\Delta_M$ over routes induces
$\bar y\in\Delta_{|\mathcal H|}$ via
\begin{align}
\bar y(h)
=
\sum_{j:\,\pi(j)=h} y_j.
\label{eq:prefix_map}
\end{align}

The feedback value is
\begin{align}
v^{\mathrm{fb}}
=
\max_{z\in\Delta_N}
\min_{y\in\Delta_M}
\bar y^\top A^{\mathrm{fb}} z.
\label{eq:v_fb}
\end{align}

\medskip
Recall that $A^{\mathrm{switch}}\in\mathbb{R}^{M\times N}$ is indexed by
routes $j\in\{1,\dots,M\}$, whereas
$A^{\mathrm{fb}}\in\mathbb{R}^{|\mathcal H|\times N}$ is indexed by
prefixes $h\in\mathcal H$.
To compare the two games, define the matrix
$\tilde A^{\mathrm{fb}}\in\mathbb{R}^{M\times N}$ by evaluating
$A^{\mathrm{fb}}$ at the prefix generated by each route. Therefore, for all $j\in\{1,\dots,M\}$ and $ i\in\mathcal N$
\begin{align}
\tilde A^{\mathrm{fb}}(j,i)
:=
A^{\mathrm{fb}}(h_j,i)
=
A^{\mathrm{fb}}(\pi(j),i).
\label{eq:Afb_tilde_def}
\end{align}
Thus $\tilde A^{\mathrm{fb}}$ simply re-indexes the rows of
$A^{\mathrm{fb}}$ by routes.

\medskip

Let $A^{\mathrm{switch}}(c)$ and $\tilde A^{\mathrm{fb}}(c)$ denote the switching
and feedback payoff matrices obtained with switching cost $c$.
Let $v^{\mathrm{switch}}(c)$ and $v^{\mathrm{fb}}(c)$ denote the corresponding
game values.

\smallskip

\begin{proposition}[Effect of seeker awareness]
\label{lem:fb_vs_switch}
Let $A^{\mathrm{switch}}(c), \,\tilde A^{\mathrm{fb}}(c)\in\mathbb{R}^{M\times N}$
be the payoff matrices defined in
Eqs.~\eqref{eq:Aswitch_restricted} and~\eqref{eq:Afb_tilde_def} for a fixed switching cost $c$.
Then the corresponding game values satisfy
\begin{align}
\begin{split}
   v^{\mathrm{fb}}(c) 
\le
v^{\mathrm{switch}}(c)
\le
 v^{\mathrm{fb}}(c) + \delta,
\end{split}
\label{eq:value_bound_fb_switch}
\end{align}
where $\delta = 
\max_{j,i}
\left|
A^{\mathrm{switch}}(j,i,c)
-
\tilde A^{\mathrm{fb}}(j,i,c)
\right|. $
\end{proposition}



\smallskip

\begin{proof}
We first express the feedback value using $\tilde A^{\mathrm{fb}}(c)$.
For any $y\in\Delta_M$, the induced distribution over prefixes defined in
Eq.~\eqref{eq:prefix_map} satisfies
\[
\sum_{h\in\mathcal H}\bar y(h)\,A^{\mathrm{fb}}(h,\cdot,c)
=
\sum_{j=1}^M y_j\,\tilde A^{\mathrm{fb}}(j,\cdot,c).
\]
Hence
\[
v^{\mathrm{fb}}(c)
=
\max_{z\in\Delta_N}
\min_{y\in\Delta_M}
y^\top \tilde A^{\mathrm{fb}}(c) z.
\]

Observe that for any route $r_j$ and location $i\in\mathcal N$,
\[
\tilde A^{\mathrm{fb}}(j,i,c)
=
A^{\mathrm{fb}}(h_j,i,c)
\le
A^{\mathrm{switch}}(j,i,c).
\]
Indeed, in the feedback sub-game, fixing the Seeker's
strategy to route $r_j\in\mathcal I(h_j)$ is always feasible.
Therefore, minimizing over all mixed strategies on
$\mathcal I(h_j)$ cannot yield a larger value than fixing the
continuation to $r_j$. Consequently,
\[
v^{\mathrm{fb}}(c)\le v^{\mathrm{switch}}(c).
\]
Let
\[
\delta =
\max_{j,i}
\left|
A^{\mathrm{switch}}(j,i,c)
-
\tilde A^{\mathrm{fb}}(j,i,c)
\right|.
\]
By standard monotonicity and translation properties of zero-sum
matrix games (see Lemma 4.3.4 in \cite{filar2012competitive}),
it follows that
\[
v^{\mathrm{fb}}(c)
\le
v^{\mathrm{switch}}(c)
\le
v^{\mathrm{fb}}(c) +\delta,
\]
which proves Eq. \eqref{eq:value_bound_fb_switch}.

\end{proof}

\begin{remark}
The bound in Proposition~\ref{lem:fb_vs_switch} depends on the switching cost
$c$ through the matrix $A^{\mathrm{switch}}(c)$. In general, increasing $c$ reduces the Hider's flexibility at the reveal stage and can therefore reduce the gap between $v^{\mathrm{switch}}$ and $v^{\mathrm{fb}}$. However, a large switching cost does not by itself imply that the two values coincide. Even when switching is too expensive to be used, the seeker-aware model may
still yield a smaller value because, after the reveal, the Seeker can select a prefix-consistent route that reaches the current hiding location more efficiently than the originally committed route. The two values coincide only in the special case where, for every revealed prefix and admissible hiding location, the committed route is already as favorable to the Seeker as any prefix-consistent alternative.
\end{remark}

%

\section{Numerical Examples}

We illustrate the restricted-seeker model on a three-location instance. The locations are numbered as $i = \{1,\,2,\,3\}$.
Let the candidate locations be
\begin{align}
1 \equiv [1,0], \qquad
2 \equiv [2,1], \qquad
3 \equiv [2,-1],
\end{align}
with the Seeker starting from the origin $O=[0,0]$.
All $M=3! = 6$ permutations of $\{1,2,3\}$ are admissible routes:
\begin{align*}
r_1&=(1,2,3),\quad r_2=(1,3,2),\quad r_3=(2,1,3),\\
r_4&=(2,3,1),\quad r_5=(3,1,2),\quad r_6=(3,2,1).
\end{align*}

The payoff is the positive cumulative travel distance accrued by the Seeker until the treasure is found. Hence the Hider maximizes and the Seeker minimizes.

\subsubsection{Base payoff matrix}
The baseline matrix $A \in \mathbb{R}^{M\times N}$, defined in Eq. \eqref{eq:baseA}, has rows indexed by routes and columns indexed by hiding locations:
\begin{align}
\label{eq:numbaseA}
A =
\begin{bmatrix}
1.0000 & 2.4142 & 4.4142\\
1.0000 & 4.4142 & 2.4142\\
3.6503 & 2.2361 & 5.0645\\
5.6503 & 2.2361 & 4.2361\\
3.6503 & 5.0645 & 2.2361\\
5.6503 & 4.2361 & 2.2361
\end{bmatrix}.
\end{align}
Each entry $A(j,\, i)$ equals the cumulative distance traveled along route $r_j$ until location $i$ is reached. The value corresponding to this base payoff matrix is obtained as $v^{\mathrm{base}} = 3.3251$.

\subsection{Restricted seeker model}
Fix the reveal time $t_{\mathrm{reveal}}=1$.
For route $r_j$, the revealed prefix is $h_j=(r_{j,1})$.
Therefore,
\begin{align}
\mathcal V(h_j)=\{r_{j,1}\}, \qquad
\mathcal U(h_j)=\mathcal N \setminus \mathcal V(h_j).
\end{align}
If $i \in \mathcal V(h_j)$, the game terminates at $t_{\mathrm{reveal}}$.
If $i \in \mathcal U(h_j)$, the Hider may either stay at $i$ or switch once to the other unvisited location $\hat{i} \in \mathcal U(h_j)$ by paying cost $c\ge 0$. Since $N=3$ and $t_{\mathrm{reveal}}=1$, each admissible set $\mathcal U(h_j)$ contains exactly two elements.

\subsubsection{Hider's switching payoff matrix}
Under the restricted-seeker model, the Seeker remains committed to route $r_j$. Thus $A^{\mathrm{switch}}$ is obtained according to Algorithm \ref{alg:Aswitch}.
For $c=1$, the resulting switching matrix is
\begin{align}
A^{\mathrm{switch}} =
\begin{bmatrix}
1.0000 & 3.4142 & 4.4142\\
1.0000 & 4.4142 & 3.4142\\
4.0645 & 2.2361 & 5.0645\\
5.6503 & 2.2361 & 4.6503\\
4.0645 & 5.0645 & 2.2361\\
5.6503 & 4.6503 & 2.2361
\end{bmatrix}.
\end{align}
\begin{figure}[!t]
    \centering   \includegraphics[width=\linewidth]{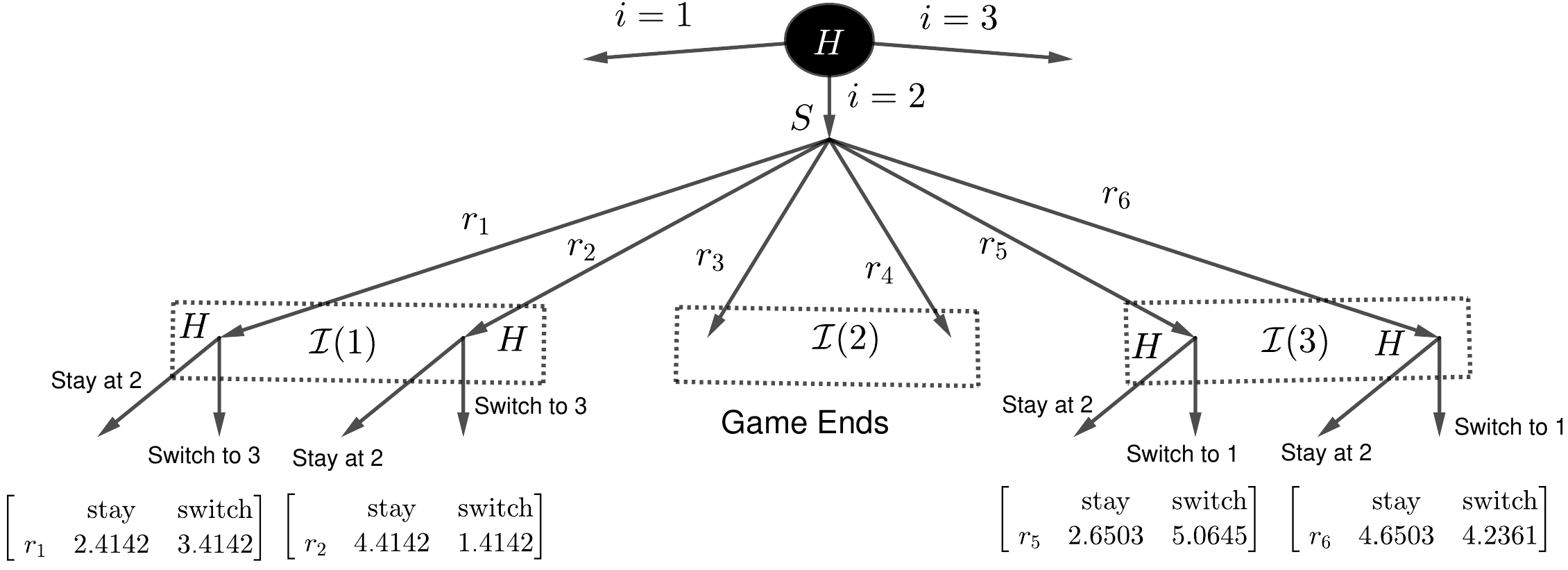}
   \caption{Restricted seeker model game tree for $N=3$ with $t_{\mathrm{reveal}}=1$ and $c=1$. After the first visit, $h_1=\{1\}$, so $\mathcal{U}(h_1)=\{2,3\}$. The Hider compares staying at 2 with switching to 3 (cost $c$), yielding $A^{\mathrm{switch}}(1,2)=3.4142$.}
    \label{fig:num1}
\end{figure}
As a concrete entry, consider $r_1=(1,2,3)$ and initial location $i=2$, as shown in Fig. \ref{fig:num1}.
At reveal time the visited set is $\{1\}$ and the unvisited set is $\{2,3\}$, so the Hider compares
\begin{align*}
\text{stay: } & A(1,2)=2.4142,\\
\text{switch to }3: ~~& A(1,3)-c=4.4142-1=3.4142.
\end{align*}
Since switching yields the larger value,
hence $A^{\mathrm{switch}}(1,2)=3.4142$.

The restricted-seeker value at $t_{\mathrm{reveal}}=1$ is
\begin{equation}
v^{\mathrm{switch}} = 
\max_{z\in\Delta_N}
\min_{y\in\Delta_M}
y^\top A^{\mathrm{switch}} z.
\end{equation}

For $c=1$, solving the above gives
\begin{align*}
v^{\mathrm{switch}} &\approx 3.6462 > v^{\mathrm{base}}\\
y^\star &\approx [0,\, 0.4310,\, 0,\, 0.3738,\, 0,\, 0.1953]^\top,\\
z^\star &\approx [0.0920,\, 0.4540,\, 0.4540]^\top.
\end{align*}

\smallskip

\subsubsection{Critical switching cost}

Fix $t_{\mathrm{reveal}}=1$ and route $r_1=(1,2,3)$.
Then
$
h_1=(1),
\qquad
\mathcal U(h_1)=\{2,3\}.
$
Consider the initial hiding location $i=2$. To evaluate
Eq.~\eqref{eq:cstar_restricted}, we compute the route-wise reduced
payoffs along $r_1$ with zero switching cost. From
Eq.~\eqref{eq:Ared_restricted},
\begin{align*}
A_1^{\mathrm{red}}(h_1,2,2,0)
&=
A(1,2)-A(1,1)\\
&=
2.4142-1.0000
=
1.4142,
\\[2mm]
A_1^{\mathrm{red}}(h_1,2,3,0)
&=
A(1,3)-A(1,1)\\
&=
4.4142-1.0000
=
3.4142.
\end{align*}

Therefore,
\begin{align*}
c^\star(1,2)
&=
\max_{\hat i\in\{2,3\}}
\Big(
A_1^{\mathrm{red}}(h_1,2,\hat i,0)
-
A_1^{\mathrm{red}}(h_1,2,2,0)
\Big)
\\
&=
\max\{0,\;3.4142-1.4142\}
\\
&=
{2.0000}.
\end{align*}
Equivalently, using the simplified expression,
\[
c^\star(1,2)
=
A(1,3)-A(1,2)
=
4.4142-2.4142
=
2.0000.
\]







{
\begin{table}[!t]
\centering
\caption{Route-wise thresholds $c^\star(j,i)$ for the three-location
instance with $t_{\mathrm{reveal}}=1$, computed using
Eq.~\eqref{eq:cstar_restricted}. ``--'' indicates
$i\in\mathcal V(h_j)$.}
\label{tab:cstar_all}
\begin{tabular}{c | c | c | c}
\hline
Route $r_j$
& $c^\star(j,1)$
& $c^\star(j,2)$
& $c^\star(j,3)$\\
\hline
$r_1=(1,2,3)$ & --     & 2.0000 & 0      \\
$r_2=(1,3,2)$ & --     & 0      & 2.0000 \\
$r_3=(2,1,3)$ & 1.4142 & --     & 0      \\
$r_4=(2,3,1)$ & 0      & --     & 1.4142 \\
$r_5=(3,1,2)$ & 1.4142 & 0      & --     \\
$r_6=(3,2,1)$ & 0      & 1.4142 & --     \\
\hline
\end{tabular}
\end{table}}
The global threshold $c^\star_{\mathrm{global}}$ is defined in
Eq.~\eqref{eq:cstar_global_restricted} as the maximum free switching advantage
over all admissible route–location pairs. 
For this instance, the largest gain occurs at prefix $h=1$ for $i=3$, yielding $c^\star_{\mathrm{global}}=2.0000$.

\smallskip

{
\subsection{Expected VOI}
We evaluate the value-of-information using
Definition~\ref{def:voi}. For $t_{\mathrm{reveal}}=1$ and $N=3$,
the admissible set $\mathcal U(h_j)$ contains exactly two nodes.
For each supported route $r_j$ and initial location
$i\in\mathcal U(h_j)$, the route-level value-of-information is
\begin{align}
\mathrm{VOI}(r_j,i;c)
={}&
\max_{\hat i\in\mathcal U(h_j)}
\min_{k\in\mathcal I^\star(h_j)}
A_k^{\mathrm{red}}(h_j,i,\hat i,c)
\nonumber\\
&-
\min_{k\in\mathcal I^\star(h_j)}
A_k^{\mathrm{red}}(h_j,i,i,c).
\end{align}
Thus, the Hider compares the worst-case payoff obtained from staying
at $i$ with that obtained from switching to the single alternative
location, over all supported routes consistent with the revealed
prefix.}

{
For this instance, the baseline Seeker strategy has
$
\operatorname{supp}(y^\star)=\{r_1,\,r_2,\,r_4,\,r_6\}.
$
At $c=1$, the corresponding route-level VOI values are
\begin{align*}
\mathrm{VOI}(r_1,:)
&=
[\,0,\;0,\;0\,],\\
\mathrm{VOI}(r_2,:)
&=
[\,0,\;0,\;0\,],\\
\mathrm{VOI}(r_4,:)
&=
[\,0,\;0,\;0.4142\,],\\
\mathrm{VOI}(r_6,:)
&=
[\,0,\;0,\;0.4142\,].
\end{align*}
Consequently, the location-wise worst-case quantity in
Eq.~\eqref{eq:barvoi_restricted}, evaluated over the supported routes, is
\[
\overline{\mathrm{VOI}}
=
\begin{bmatrix}
0 & 0 & 0
\end{bmatrix}^{\top} \implies 
\mathbb E[\mathrm{VOI}]
=
(z^\star)^\top\overline{\mathrm{VOI}}
= 0
\]}



\smallskip

\subsubsection{Illustration with $N=6$}

We consider an instance with $N=6$ candidate locations
\begin{align*}
p_1&=\begin{bmatrix}1 & 1\end{bmatrix}^\top,\;
p_2=\begin{bmatrix}2 & 2\end{bmatrix}^\top,\;
p_3=\begin{bmatrix}2 & 1\end{bmatrix}^\top,\\
p_4&=\begin{bmatrix}5 & 1\end{bmatrix}^\top,\;
p_5=\begin{bmatrix}3 & 5\end{bmatrix}^\top,\;
p_6=\begin{bmatrix}5 & 3\end{bmatrix}^\top ,
\end{align*}
as shown in Fig.~\ref{fig:pointsN6}, with the Seeker starting from
$O=[0,0]^\top$. The route set consists of all permutations of
$\{1,\dots,6\}$, so $M=6!=720$.

The baseline payoff matrix $A\in\mathbb R^{M\times N}$ is computed using
Eq.~\eqref{eq:baseA}. We fix $t_{\mathrm{reveal}}=1$ and $c=1$.
For each realized route $r_j$, the prefix is $h_j=(r_{j,1})$,
so $\mathcal V(h_j)=\{r_{j,1}\}$ and
$\mathcal U(h_j)=\mathcal N\setminus\mathcal V(h_j)$.
Using Algorithm~\ref{alg:Aswitch}, we construct
$A^{\mathrm{switch}}\in\mathbb R^{M\times N}$.
{For this instance, we obtain
\[
v^{\mathrm{base}}\approx 8.0276,
\qquad
v^{\mathrm{switch}}\approx 10.7737.
\]
The difference in these values demonstrates the performance gain enabled by the reveal-stage relocation option.}

\begin{figure}
    \centering
    \includegraphics[width=0.75\linewidth]{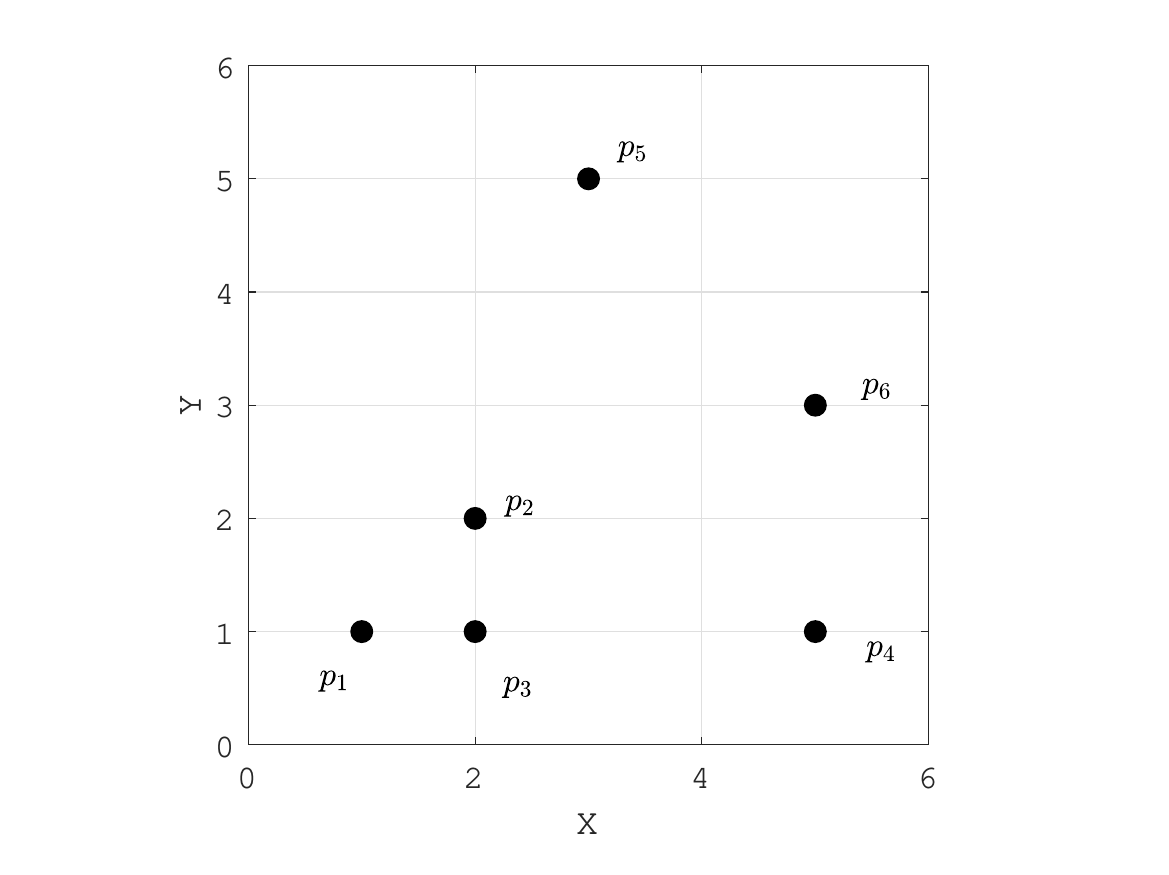}
    \caption{Six-location configuration used in the numerical example, with the Seeker starting at the origin.}
    \label{fig:pointsN6}
\end{figure}

Given $A^{\mathrm{switch}}$, the route-level value-of-information
$\mathrm{VOI}(r_j,i)$ is computed directly from
Definition~\ref{def:voi} by enumerating admissible
$\hat i\in\mathcal U(h_j)$ for each $(r_j,i)$.
We then evaluate $\mathbb{E}[\mathrm{VOI}]$
under equilibrium strategies and repeat this computation over
varying switching costs $c$ and reveal times $t_{\mathrm{reveal}}$.
\begin{figure}
    \centering
    \includegraphics[width=0.9\linewidth]{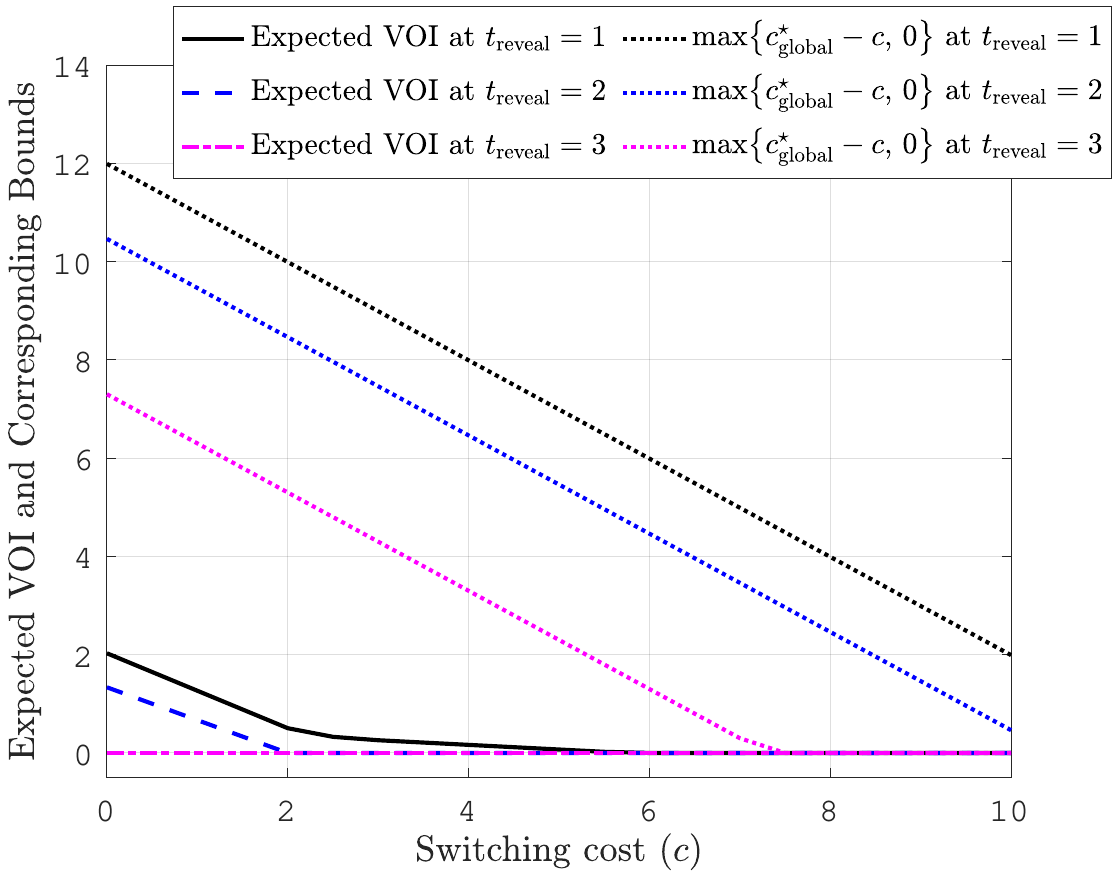}
   \caption{Expected value-of-information $\mathbb{E}[\mathrm{VOI}]$ versus switching cost $c$ for different reveal times $t_{\mathrm{reveal}}$. The decay in $c$ and downward shift with increasing $t_{\mathrm{reveal}}$ are consistent with Theorem~\ref{lem:expvoi_bound} and Proposition~\ref{lem:treveal_mono}.}
    \label{fig:voiplot}
\end{figure}
Fig. \ref{fig:voiplot} shows that, for each fixed
$t_{\mathrm{reveal}}$, the expected VOI decreases with increasing $c$
and becomes zero once $c$ exceeds the corresponding threshold,
in agreement with Theorem~\ref{lem:expvoi_bound}.
Moreover, the curves shift downward as $t_{\mathrm{reveal}}$ increases and
later reveals a reduction in the admissible switching set and hence diminishes
the informational advantage, consistent with
Lemma~\ref{lem:treveal_mono}.
In this instance, $t_{\mathrm{reveal}}=1$ yields the largest
information gain and exhibits a near-linear decay in $c$,
whereas for $t_{\mathrm{reveal}}> 2$ the expected VOI is zero over the plotted range.

\subsection{Seeker-aware Model}

We now consider the same three-location instance under the seeker-aware
(feedback) model with $t_{\mathrm{reveal}}=1$. The geometry, routes,
and baseline matrix $A$ remain unchanged from the previous subsection.

\subsubsection{Feedback payoff matrix}
At $t_{\mathrm{reveal}}=1$, the revealed prefix is $h_j=r_{j,1}$.
Routes sharing the same first node form the information sets
\begin{align}
\mathcal I(1)=\{r_1,r_2\},~
\mathcal I(2)=\{r_3,r_4\},~
\mathcal I(3)=\{r_5,r_6\}.
\end{align}
{
In the feedback model, for each revealed prefix $h$ and initial hiding location
$i\in\mathcal U(h)$, we construct the corresponding continuation
payoff matrix over the prefix-consistent routes in $\mathcal I(h)$
and the admissible post-reveal locations in $\mathcal U(h)$.
Accordingly, the value of the
corresponding zero-sum sub-game determines the entry
$A^{\mathrm{fb}}(h,i)$.  If $i\in\mathcal V(h)$, the game has already
terminated at the reveal stage and the corresponding entry is given by
the cumulative payoff $A(j,i)$. For $c=1$, this yields
\begin{align}
A^{\mathrm{fb}}
=
\begin{bmatrix}
1 & 2.9142 & 2.9142\\
3.9432 & 2.2361 & 4.3574\\
3.9432 & 4.3574 & 2.2361
\end{bmatrix},
\end{align}
where rows correspond to prefixes $h=1,2,3$ and columns correspond
to hiding locations $i=1,2,3$.}

\begin{figure}
    \centering
    \includegraphics[width=\linewidth]{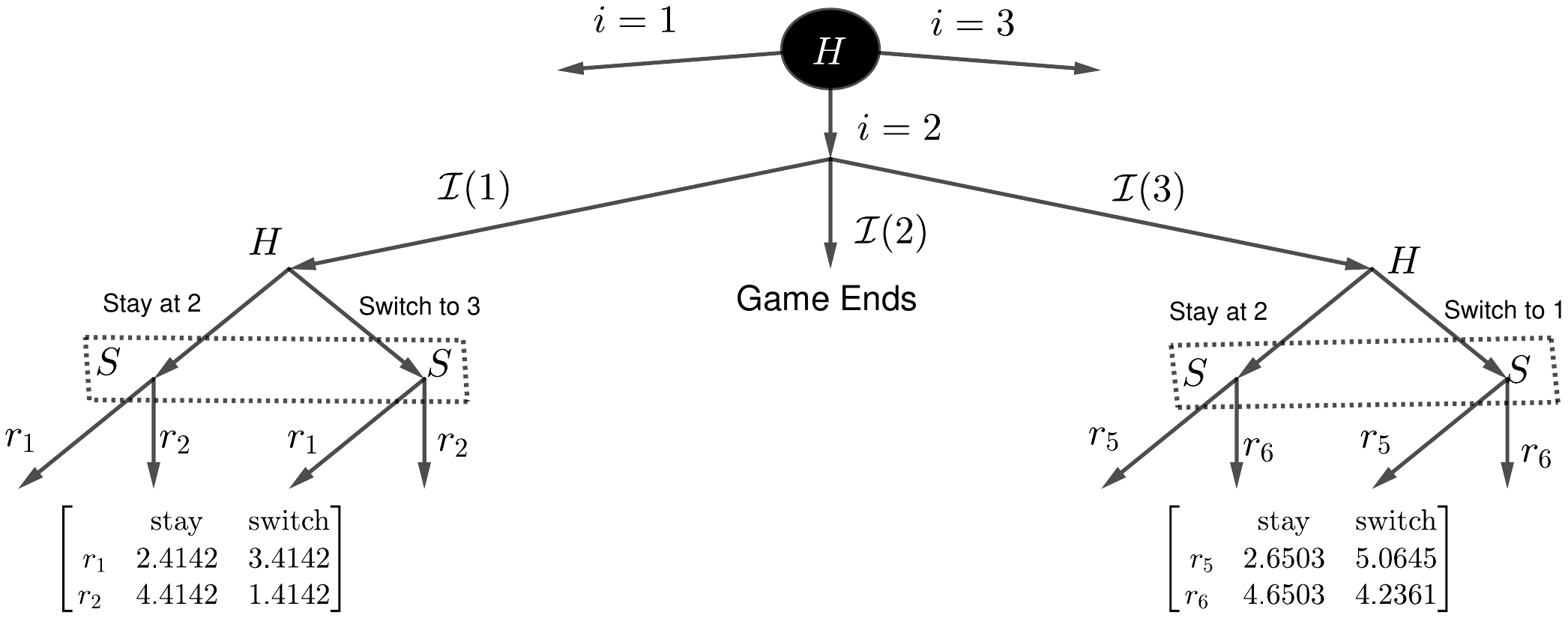}
    \caption{{Feedback-model game tree at $t_{\mathrm{reveal}}=1$, $c=1$,
for prefix $h=1$ and $i=2$. Solving the zero-sum sub-game
over the prefix-consistent routes and admissible stay/switch actions
yields $A^{\mathrm{fb}}(1,2)=2.9142
<A^{\mathrm{switch}}(1,2)=3.4142$.}}
    \label{fig:num2}
\end{figure}
As an illustration, consider prefix $h=1$ (routes $r_1$ and $r_2$)
and initial location $i=2$.
The Hider compares staying at $2$ with switching to $3$
while paying $c=1$. {The corresponding
sub-game is formed using the prefix-consistent routes
$r_1$ and $r_2$ together with these two admissible Hider actions.
Solving this zero-sum sub-game for $c=1$ gives
$
A^{\mathrm{fb}}(1,2)=2.9142
<
A^{\mathrm{switch}}(1,2)=3.4142.
$}
To compare directly with the restricted matrix
$A^{\mathrm{switch}}\in\mathbb R^{M\times N}$,
we evaluate $A^{\mathrm{fb}}$ at the prefix generated by each route and define
\begin{align}
\tilde A^{\mathrm{fb}}(j,i)
:=
A^{\mathrm{fb}}(h_j,i),
\qquad j=1,\dots,6.
\end{align}
Thus $\tilde A^{\mathrm{fb}}\in\mathbb R^{M\times N}$,
and routes sharing the same first visited node have identical rows.
For $c=1$, this gives
\begin{align}
\tilde A^{\mathrm{fb}}
=
\begin{bmatrix}
1 & 2.9142 & 2.9142\\
1 & 2.9142 & 2.9142\\
3.9432 & 2.2361 & 4.3574\\
3.9432 & 2.2361 & 4.3574\\
3.9432 & 4.3574 & 2.2361\\
3.9432 & 4.3574 & 2.2361
\end{bmatrix}.
\end{align}

The feedback value is computed from Eq.~\eqref{eq:v_fb}.
For $c=1$, we obtain
\begin{align*}
v^{\mathrm{fb}} &= 2.9142,\\
\bar{y}^{\star}_{\mathrm{fb}} &\approx [1,\,0,\,0]^\top,\\
z^{\star}_{\mathrm{fb}} &\approx [0,\,0.3197,\,0.6803]^\top.
\end{align*}

Since $\tilde A^{\mathrm{fb}}(j,i)\le A^{\mathrm{switch}}(j,i)$
for all $j,i$, Proposition~\ref{lem:fb_vs_switch} implies
$
v^{\mathrm{fb}}\le v^{\mathrm{switch}},
$
which is consistent with the theoretical result established earlier.

\subsubsection{Verification of the value bound}

We next verify Proposition~\ref{lem:fb_vs_switch} for the
three-location instance with $t_{\mathrm{reveal}}=1$ and $c=1$.
Using the matrices $A^{\mathrm{switch}}$ and $\tilde A^{\mathrm{fb}}$
computed above, we obtain
\begin{align}
\left|A^{\mathrm{switch}}-\tilde A^{\mathrm{fb}}\right|
=
\begin{bmatrix}
0 & 0.5000 & 1.5000\\
0 & 1.5000 & 0.5000\\
0.1213 & 0 & 0.7071\\
1.7071 & 0 & 0.2929\\
0.1213 & 0.7071 & 0\\
1.7071 & 0.2929 & 0
\end{bmatrix}.
\end{align}
Hence,
\begin{align}
\delta
:=
\max_{j,i}
\left|
A^{\mathrm{switch}}(j,i)-\tilde A^{\mathrm{fb}}(j,i)
\right|
=
1.7071.
\end{align}
The maximum is attained at $(j,i)=(4,1)$ and $(j,i)=(6,1)$.

From the numerical computations above,
\begin{align*}
v^{\mathrm{switch}}
&= 3.6462
\le
2.9142 + 1.7071
= 4.6213.
\end{align*}
Thus, the numerical values satisfy the bound
\begin{align*}
v^{\mathrm{fb}}
\le
v^{\mathrm{switch}}
\le
v^{\mathrm{fb}}+\delta,
\end{align*}
which is consistent with Proposition~\ref{lem:fb_vs_switch}.

We also examine the behavior of the bound for large switching cost.
For sufficiently large $c$ (for example $c=100$), switching is never
beneficial for the Hider, so the restricted model reduces to the
baseline game and
\[
v^{\mathrm{switch}} = v^{\mathrm{base}} = 3.3251.
\]
{In contrast, for sufficiently large $c$, switching is suboptimal in
each feedback subgame, so the Hider optimally remains at
the initial hiding location. The continuation-game value therefore
reduces to the minimum payoff over the prefix-consistent routes,
yielding}
$
v^{\mathrm{fb}} = 2.4142.
$
The maximum entrywise difference between
$A^{\mathrm{switch}}$ and $\tilde A^{\mathrm{fb}}$ is
$
\delta = 2,
$
and remains unchanged for larger values of $c$.
Therefore the bound in Proposition~\ref{lem:fb_vs_switch}
\[
v^{\mathrm{fb}}
\le
v^{\mathrm{switch}}
\le
v^{\mathrm{fb}}+\delta
\]
continues to hold, with $v^{\mathrm{switch}} = v^{\mathrm{base}}$.

\subsection{Discussions}
The numerical study illustrates how partial route information alters the strategic evolution of the game. When switching is inexpensive, the Hider can exploit the revealed prefix to relocate to a more favorable unvisited location, leading to a strictly larger payoff under the restricted model. As the switching cost increases, this flexibility diminishes, and the value-of-information decreases accordingly. Once $c \ge c^\star_{\mathrm{global}}$, relocation is no longer advantageous, and the game reduces to the baseline case. The timing of the reveal also plays a critical role. An earlier reveal enlarges the admissible set $\mathcal U(h)$ and provides greater opportunity for improvement, whereas a later reveal restricts feasible relocation and limits the attainable gain.

Under the seeker-aware model, the Seeker anticipates this relocation and minimizes over feasible continuations within each information set. This reduces the Hider's benefit relative to the restricted case. Overall, strategic awareness and the timing of revelation jointly determine the quantitative impact of partial information on equilibrium outcomes.

\section{Conclusion}

This paper {introduced} a variation of the classic hide-and-seek game in which the Seeker’s route is partially revealed during execution. After observing the revealed prefix, the Hider may relocate once by paying a switching cost. We quantified the VOI created by this partial information under restricted seeker model. In the restricted model, the Seeker commits to a route in advance. We also analysed the seeker-aware model, in which, the Seeker may choose a prefix-consistent route after the reveal.
Our analysis characterized how the expected VOI depends on the switching cost and the reveal time. In particular, we showed that the informational advantage decreases as the switching cost increases and as the reveal occurs later along the route. Numerical examples illustrated these trends and demonstrated how partial information can change the equilibrium outcomes of the game.

Future work will focus on improving the computational scalability of the framework for larger environments and extending the model to more general sensing, relocation, and multi-agent settings.
  \bibliographystyle{ieeetr}
 \bibliography{cdc-ref}
\end{document}